\documentclass{amsart}

\newtheorem{lemma}{Lemma}[section]

\makeatletter
\CheckCommand*{\@cite}[2]{%
  {%
    \@citestyle[\citeform{#1}\if@tempswa, #2\fi]%
  }%
}
\renewcommand*{\@cite}[2]{%
  {%
    \@citestyle\citeform{#1}\if@tempswa, #2\fi
  }%
}
\makeatother

\author{Elvis Barakovic}
\address{
Department of Mathematics \newline \indent
Faculty of Science and Mathematics \newline \indent
University of Tuzla  \newline \indent
Bosnia and Herzegovina}
\email{elvis.barakovic@untz.ba}
\urladdr{http://pmf.untz.ba/staff/elvis.barakovic/}
\author{Vedad Pasic}
\address{
Department of Mathematics \newline \indent
Faculty of Science and Mathematics \newline \indent
University of Tuzla  \newline \indent
Bosnia and Herzegovina}
\email{vedad.pasic@untz.ba}
\urladdr{http://pmf.untz.ba/vedad/}

\title{Physical interpretation of pp-waves with axial torsion}

\keywords{Metric-affine gravity; Einstein--Weyl theory; PP-waves; Torsion waves}

\begin{document}
\begin{abstract}
We consider generalised pp-waves with purely axial torsion, which we previously showed to be new vacuum solutions of quadratic metric-affine gravity. Our analysis shows that classical pp-waves of parallel Ricci curvature should not be viewed on their own. They are a particular representation of a wider class of solutions, namely generalised pp-waves of parallel Ricci curvature. We compare our pp-waves with purely axial torsion to solutions of Einstein-Weyl theory, the classical model describing the interaction of gravitational and massless neutrino fields.
\end{abstract}
\maketitle

\section{Introduction}\label{introduction}
Spacetime is considered to be a connected real 4-manifold $M$ equipped with a Lorentzian metric $g$
and an affine connection $\Gamma$.
This approach, where the connection is viewed independently from the metric is  called \emph{metric-affine gravity}.
In \emph{quadratic} metric-affine gravity, we define the action as
$ \displaystyle
S:=\int q(R)
$
where $q$ is an $O(1, 3)$--invariant quadratic form on curvature $R$. Independently varying the action with respect to the metric $g$ and the connection $\Gamma$ produces the system of Euler--Lagrange equations which we will write symbolically as
\begin{eqnarray}
\label{eulerlagrangemetric}
\partial S/\partial g&=&0\\
\label{eulerlagrangeconnection}
\partial S/\partial \Gamma&=&0.
\end{eqnarray}
We consider a pp-wave as a Riemannian spacetime which admits a parallel spinor field. Classical pp-waves of parallel Ricci curvature were shown to be solutions of \eqref{eulerlagrangemetric}, \eqref{eulerlagrangeconnection} by Vassiliev.\textsuperscript{\cite{vassiliev2002pseudoinstantons,vassiliev2004quadratic}}
In our previous paper  Ref.~\cite{pasic2015torsion}, we introduced  generalised
pp-waves with purely axial torsion  as metric compatible spacetimes with pp-metric and torsion $\label{define torsion} T :=  *A$, where $A$ is a real vector field defined by $A = k(\varphi) l$, where  $l$ is a real parallel null lightlike vector and $k:\mathbb{R}\mapsto \mathbb{R}$ is an arbitrary real function of the phase $\varphi:M\mapsto\mathbb{R},\  \varphi(x) := \int l\cdot dx$. If we were to write down the pp-metric locally as
\begin{equation}\label{explicit pp metric}
\,d s^2= \,2\,d x^0\,d x^3-(d x^1)^2-(d x^2)^2 +f(x^1,x^2,x^3)\,(d x^3)^2,
\end{equation}
in some local coordinates $(x^0,x^1,x^2,x^3)$,  and for
\begin{equation}\label{explicit l and a}
l^\mu=(1,0,0,0),\ m^\mu=(0,1,\mp \mathrm{i},0)
\end{equation}
we get that $\varphi(x)=x^3+\mathrm{const}$. The torsion $T$ is purely axial and the connection of a generalised pp-wave with purely axial torsion is metric compatible.
We have shown that generalised pp-waves with purely axial torsion of parallel $\{\!Ric\!\}$ are
solutions of \eqref{eulerlagrangemetric}, \eqref{eulerlagrangeconnection} in the Yang--Mills case.  The remarkable property  is that the curvature of a generalised pp-wave is a
sum of the curvature of the underlying classical pp-space
\begin{equation}
\label{curvature of the underlying classical pp-space}
-\frac12(l\wedge\{\!\nabla\!\})\otimes(l\wedge\{\!\nabla\!\})f
\end{equation}
 and the curvature
\begin{equation}
\label{curvature generated by a torsion wave} \frac{1}{4}k(\varphi)^2\mathrm{Re} \left((l\wedge m)\otimes(l\wedge \overline{m})\right)
\mp\frac{1}{2}k'(\varphi)\mathrm{Im} \left((l\wedge m)\otimes(l\wedge \overline{m})\right)
\end{equation}
generated by a axial torsion wave traveling over the pp-space.
Ricci curvature is
\begin{equation}\label{RicEW}
Ric = \frac12\left(f_{11}+f_{22}-k^2\right) (l\otimes l),
\end{equation}
where $f_{\alpha\beta} = \partial_{\alpha}\partial_{\beta} f$ and scalar curvature $\mathcal{R}$ is equal to zero. Similarly, the property that curvatures \eqref{curvature of the underlying classical pp-space} and \eqref{curvature generated by a torsion wave} add up was also present in the case of generalised pp-waves with purely tensor torsion, see Refs.~\cite{pasic2014pp,pasic2005pp}.
In our previous paper  Ref.~\cite{pasic2014pp}, we gave the physical interpretation of generalised pp-waves with purely tensor torsion constructed in Ref.~\cite{pasic2005pp}. Similarly to the approach of Ref.~\cite{pasic2014pp}, now we want to  compare the generalised pp-waves with purely axial torsion to the solutions of the
classical models describing the interaction of gravitational and massless neutrino
fields, namely Einstein--Weyl theory.

Our torsion and torsion generated curvature can be interpreted as waves traveling at speed of light. The underlying classical pp-space of parallel
Ricci curvature can then be viewed as the gravitational imprint
created by a wave of some massless matter field. As pointed out in  Ref.~\cite{pasic2014pp}, such a situation
occurs in Einstein--Weyl theory. We choose to deal with the complexified curvature
$$
\mathfrak{R}:=r\,(l\wedge m)\otimes(l\wedge \overline{m}),
$$
where $r:=\frac14 k^2- \frac{\mathrm{i}}{2} k'$. Note that the function  $r$ is a function of the phase
$\varphi$ and  the curvature \eqref{curvature generated by a torsion wave} generated by the axial torsion is equal to $\textrm{Re}(\mathfrak{R})$.  The  curvature  $\mathfrak{R}$ is polarized, i.e. ${}^{*}\mathfrak{R}=-\mathfrak{R}^{*}=\pm\textrm{i}\mathfrak{R}$, and it can be written as
\begin{equation}
\label{spinor representation of curvature}
\mathfrak{R}_{\alpha\beta\gamma\delta} =\sigma_{\alpha\beta
ab}\,\omega^{abcd}\,\overline{\sigma}_{\gamma\delta cd},
\end{equation}
where $\omega$ is some symmetric rank 4 spinor and $\sigma_{\alpha\beta}$ are second order Pauli matrices defined by
$\sigma_{\alpha\beta ac}:=\frac12
\bigl( \sigma_{\alpha a\dot b}\epsilon^{\dot b\dot d}\sigma_{\beta
c\dot d} - \sigma_{\beta a\dot b}\epsilon^{\dot b\dot
d}\sigma_{\alpha c\dot d} \bigr)\,$ and $\overline{\sigma}$ denotes their complex conjugation.
Resolving (\ref{spinor
representation of curvature}) with respect to~$\omega$ yields
\begin{equation}
\label{formula for omega} \omega=\xi\otimes\xi\otimes\xi\otimes\xi,
\end{equation}
where
\begin{equation}
\label{formula for xi} \xi:=r^{1/4}\,\chi
\end{equation}
and $\chi^a=\left(1,\ 0\right)$ is the parallel spinor field of the underlying pp-space. Formula (\ref{formula for omega}) shows that the rank 4 spinor
$\omega$  is the 4th tensor
power of a rank 1 spinor $\xi$. Hence, the curvature $\mathfrak{R}$ is completely determined by
the rank 1 spinor field $\xi$.
\begin{lemma}\label{LemmaSpinorSatisfiesmasslessDirac}
The spinor field (\ref{formula for xi}) satisfies
 the massless Dirac equation.
\end{lemma}
\begin{proof}
The massless Dirac equation, see Appendix B of Ref.~\cite{pasic2014pp}, is explicitly given by
$$
\sigma^\mu{}_{a\dot{b}}\nabla_\mu\xi^a-\frac12T^\eta{}_{\eta\mu}\sigma^\mu{}_{a\dot{b}}\xi^a=0
$$
which can be equivalently written as
\begin{equation}\label{massless Dirac Equation temp1}
\sigma^\mu{}_{a\dot{b}}\{\nabla\}_\mu\xi^a
\pm\frac{\mathrm{i}}{4}\varepsilon_{\alpha\beta\gamma\delta}T^{\alpha\beta\gamma}\sigma^\delta{}_{a\dot{b}}\xi^a=0,
\end{equation}
where $\{\nabla\}$ is the covariant derivative with respect to the Levi-Civita connection. For  the definition of the covariant derivative of a spinor field see Appendix A of Ref.~\cite{pasic2014pp} and Section~2.5 of Ref.~\cite{pasic2009new}. Since the classical pp-wave spacetime admits a parallel spinor field $\chi$ we have that  $\sigma^\mu{}_{a\dot{b}}\{\nabla\}_\mu\xi^a=(r^{1/4})' \sigma^\mu{}_{a\dot{b}}l_\mu \chi^a=0$ for the special local coordinates \eqref{explicit pp metric}, \eqref{explicit l and a} and for the  Pauli matrices for the pp-metric
\begin{equation}\label{Pauli matrices for pp metric}
\sigma^0{}_{a\dot{b}}=\left(\!
                         \begin{array}{cc}
                           1 & 0 \\
                           0 & -f \\
                         \end{array}\!
                       \right),
\sigma^1{}_{a\dot{b}}=\left(\!
                         \begin{array}{c}
                           0 \  1 \\
                           1 \  0 \\
                         \end{array}\!
                       \right),
\sigma^2{}_{a\dot{b}}=\left(\!
                         \begin{array}{cc}
                           0 & \mp \mathrm{i} \\
                           \pm \mathrm{i} & 0 \\
                         \end{array}\!
                       \right),
\sigma^3{}_{a\dot{b}}=\left(\!
                         \begin{array}{c}
                           0 \  0 \\
                           0 \  2 \\
                         \end{array}\!
                       \right).
\end{equation}
Also, for the  Pauli matrices \eqref{Pauli matrices for pp metric} and special local coordinates \eqref{explicit pp metric}, \eqref{explicit l and a}, we have that
\begin{align*}
\varepsilon_{\alpha\beta\gamma\delta}T^{\alpha\beta\gamma}\sigma^\delta{}_{a\dot{b}}\xi^a&=
\varepsilon_{\alpha\beta\gamma\delta}l_{\mu}k(\varphi)\varepsilon^{\mu\alpha\beta\gamma}
\sigma^\delta{}_{a\dot{b}}\xi^a
=6k(\varphi)l_\mu \sigma^\mu{}_{a\dot{b}}\xi^a=0,
\end{align*}
i.e. the massless Dirac equation \eqref{massless Dirac Equation temp1} is satisfied.
\end{proof}
\section{Einstein--Weyl Field Equations}\label{EinsteinWeylFieldEquations}
We consider the action as
\begin{equation}\label{EWaction}
S_{EW}:=2i\int \left( \xi^a\,\sigma^\mu{}_{a\dot
b}\,(\{\!\nabla\!\}_\mu\overline\xi^{\dot b}) \ -\
(\{\!\nabla\!\}_\mu\xi^a)\,\sigma^\mu{}_{a\dot
b}\,\overline\xi^{\dot b} \right) + K \int \mathcal{R},
\end{equation}
with the constant $\displaystyle K={c^4}/{16\pi G}$.
In Einstein--Weyl theory the connection is assumed to be
Levi-Civita, so we obtain the  Einstein--Weyl field equations varying the action (\ref{EWaction}) with
respect to the metric and the spinor, i.e.
\begin{eqnarray}\label{firstEL}
\partial S_{EW}/\partial g&=&0,\\
\label{secondEL}
\partial S_{EW}/\partial \xi&=&0.
\end{eqnarray}
The massless Dirac equation is obtained by varying the action (\ref{EWaction}) with respect to
the spinor. The variation of the first term of the action (\ref{EWaction}) with
respect to the metric yields the energy momentum tensor. For the detailed derivation of formula for the energy momentum tensor see Appendix B of Ref.~\cite{pasic2014pp}.
The explicit representation of the Einstein--Weyl field
equations (\ref{firstEL}), (\ref{secondEL}) is
\begin{align}
\frac{i}{2}\!\left[ \sigma^{\nu}{}_{a\dot b}
\left(\overline\xi^{\dot b}
\{\!\nabla\!\}^\mu\xi^a-\xi^a\{\!\nabla\!\}^\mu\overline\xi^{\dot b}
\right) \!+\!\sigma^{\mu}{}_{a\dot b} \left(\overline\xi^{\dot b}
\{\!\nabla\!\}^\nu\xi^a-\xi^a\{\!\nabla\!\}^\nu\overline\xi^{\dot b}
\right) \right]\quad\quad&  \nonumber
\\
\label{EWexplicit1} +i\left( \xi^a\,\sigma^\eta{}_{a\dot
b}\,(\{\!\nabla\!\}_\eta\overline\xi^{\dot b})g^{\mu\nu} -
(\{\!\nabla\!\}_\eta\xi^a)\,\sigma^\eta{}_{a\dot
b}\,\overline\xi^{\dot b}g^{\mu\nu}\right)
- K Ric^{\mu\nu}+\frac{K}{2}\mathcal{R}g^{\mu\nu}&=0,\\
\label{EWexplicit2} \sigma^\mu{}_{a\dot b}\{\! \nabla\! \}_\mu\,\xi^a
&=0.
\end{align}
\section{Comparison of Metric-affine and  Einstein--Weyl Solutions}
The examination of the Einstein--Weyl field equations has a long history, see Refs.~\cite{griffiths1970two,griffiths1970tetrad,audretsch1970neutrino,audretsch1971asymptotic,griffiths1972some,griffiths1972gravitational,collinson1973space,davis1974ghost,kuchowicz1978presence}. One review of known solutions of Einstein--Weyl theory is given in Ref.~\cite{pasic2014pp}. The
nonlinear system of equations (\ref{EWexplicit1}),
(\ref{EWexplicit2}) has solutions in the form of pp-waves.
 We wish to present a
class of explicit solutions of (\ref{EWexplicit1}),
(\ref{EWexplicit2}) where the metric $g$ is in the form of the
pp-metric and the spinor $\xi$ as in (\ref{formula for xi}).
The spinor (\ref{formula for xi}) satisfies the massless Dirac equation
(\ref{EWexplicit2}). The scalar curvature is zero
and as the spinor $\chi$ appearing in formula (\ref{formula
for xi}) is parallel, hence the equation (\ref{EWexplicit1}) now becomes
\begin{equation*}
\frac{i}{2} \sigma^{\nu}{}_{a\dot b} \left(\overline\xi^{\dot
b}\{\!\nabla\!\}^\mu\xi^a-\xi^a\{\!\nabla\!\}^\mu\overline\xi^{\dot
b}\right) +\frac{i}{2}\sigma^{\mu}{}_{a\dot b}
\left(\overline\xi^{\dot
b}\{\!\nabla\!\}^\nu\xi^a-\xi^a\{\!\nabla\!\}^\nu\overline\xi^{\dot
b}\right) - K Ric^{\mu\nu}  =  0.
\end{equation*}
We now need to determine under which conditions the above equation is satisfied.
Substituting formulae (\ref{RicEW}), (\ref{formula for xi}) into the above
equation, and using  $\nabla\chi=0$, we get that
\[
i (\sigma^{\nu}{}_{a\dot b}l^\mu +\sigma^{\mu}{}_{a\dot b}l^\nu)\!\!
\left((r^{1/4})'\ \overline{r^{1/4}}-r^{1/4}\
(\overline{r^{1/4}})'\right)\chi^a\overline{\chi}^{\dot b}=
K l^\mu l^\nu \!\left(f_{11}+f_{22}-k(x^3)^2\right).
\]
The condition that a pp-wave needs to satisfy to be a solution of Einstein--Weyl is
\begin{equation}\label{condition on s}
f_{11}+f_{22} =
 k(x^3)^2+
\frac{2i}{K}\left((r^{1/4})'\
\overline{r^{1/4}}-r^{1/4}\ (\overline{r^{1/4}})'\right),
\end{equation}
since $\sigma^\mu{}_{a\dot b}\chi^a\overline{\chi}^{\dot b}=l^\mu$.
The complex valued function $r(\varphi)$  can
be chosen arbitrarily and for the fixed function $k(\varphi)$ it uniquely determines the RHS of
(\ref{condition on s}). The main difference between the two models is that in the metric-affine model the generalised pp-wave solutions have parallel $\{\!Ric\!\}$ curvature, whereas in the Einstein--Weyl model the pp-wave type solutions do not necessarily have parallel Ricci curvature.
The comparison of this two types of solutions becomes much clearer in the case of the monochromatic solutions as was done in Ref.~\cite{pasic2014pp}.
In the metric-affine case the Laplacian of $f$ can be any constant,
while in the Einstein--Weyl case it is required for it to be a particular
constant, which is the consequence of conformal invariance of the metric-affine model and the presence of the gravitational constant in the Einstein--Weyl.
The generalised pp-waves of parallel Ricci curvature are very
similar to pp-type solutions of the Einstein--Weyl model. According this conclusion, similarly to Ref.~\cite{pasic2014pp}, we propose that generalised pp-waves with purely axial torsion and parallel Ricci curvature represent a metric-affine model for the massless neutrino.

\section*{Acknowledgments}

We thank the organisers Remo Ruffini, Robert Jantzen and Massimo Bianchi for awarding us grants in order to attend the 14th Marcel Grossmann meeting.


\label{lastpage}

\end{document}